\documentclass[journal,onecolumn]{IEEEtran}
\usepackage[utf8x]{inputenc}
\usepackage[T1]{fontenc}
\usepackage{graphics} 
\usepackage{epsfig} 
\usepackage{mathptmx} 
\usepackage{times} 
\usepackage{amsmath} 
\usepackage{amssymb}  
\newtheorem{lemma}{Lemma}
\newtheorem{proof}{Proof}
\newtheorem{theorem}{Theorem}

\usepackage{wrapfig}
\usepackage{subfig}
\usepackage{xcolor}
\usepackage{soul}
\usepackage{stix}
\usepackage{url}

\usepackage[a4paper, lmargin=0.1666\paperwidth, rmargin=0.1666\paperwidth, tmargin=0.1111\paperheight, bmargin=0.1111\paperheight]{geometry} 

\usepackage[all]{nowidow} 
\usepackage[protrusion=true,expansion=true]{microtype} 

\begin{document}
%
\title{\LARGE{\textbf{Dual-Loop Robust Control of Biased Koopman Operator Model by Noisy Data of Nonlinear Systems}}} 

\author{Tianyi He, Anuj Pal
\thanks{Tianyi He is with the Department of Mechanical and Aerospace Engineering, Utah State University, Logan, Utah 84322-4130, USA. Email: tianyi.he@usu.edu. He is the corresponding author.}
\thanks{Anuj Pal was with the Department
	of Mechanical Engineering, Michigan State University, East Lansing,
	MI, 48824 USA. E-mail: palanuj@msu.edu.}
}

%
%

\markboth{Journal of \LaTeX\ Class Files,~Vol.~14, No.~8, August~2015}%
{Shell \MakeLowercase{\textit{et al.}}: Bare Demo of IEEEtran.cls for IEEE Journals}
%



\maketitle

\begin{abstract}
The Koopman operator approach for data-driven control design of a nonlinear system is on the rise because of its capability to capture the behaviours of global dynamics. However, the measurement noises of inputs and outputs will bias the Koopman model identification and cause model mismatch from the actual nonlinear dynamics. The current work evaluates the bounds of the noise-induced model bias of the Koopman operator model and proposes a data-driven robust dual-loop control framework (\textit{Koopman based robust control}-KROC) for the biased model. First, the model mismatch is found bounded under radial basis functions (RBF) and the bounded noises, and the bound of model mismatch is assessed. Second, the pitfalls of linear quadratic Gaussian (LQG) control based on the biased Koopman model of Van Der Pol oscillator are shown. Motivated from the pitfalls, the dual-loop control is proposed, which consist of an observer-based state-feedback control based on the nominal Koopman model and an additional robust loop to compensate model mismatch. A linear matrix inequality (LMI) is derived, which can guarantee robust stability and performance under bounded noises for the finite-dimensional Koopman operator model. Finally, the proposed framework is implemented to a nonlinear Van Der Pol oscillator to demonstrate enhanced control performance by the dual-loop robust control.

\end{abstract}


\section{Introduction}

Nonlinear systems pose challenges in control theory, because their complex and nonlinear behaviours make the design of stabilizing controllers a non-trivial task~\cite{Khalilnonlinear,slotine1991applied}. Traditional model-based nonlinear control requires an accurate model using the first-principle-based approaches. To alleviate the modeling efforts, data-driven approaches to perform system identification and control design are on the rise \cite{hou2013model}.


The Koopman operator theory~\cite{koopman1931hamiltonian} has emerged as a data-driven approach to establish models and design controller for nonlinear systems using data. The Koopman operator theory involves lifting the original nonlinear dynamics to an infinite-dimensional observable space to identify its linear representation ~\cite{brunton2021modern,mauroy2020koopman,bevanda2021koopman}. After lifting, well-established linear control theories can be used to design linear controller to achieve desired performance. Compared to nonlinear control algorithms that involve solving complex nonlinear optimization problems, linear control algorithms are more efficient for real-time implementation. As a result, Koopman model-based control has gained significant attentions in various engineering applications~\cite{bruder2020data,mamakoukas2021derivative}. This approach is particularly useful in scenarios where the nonlinear system is too complex to derive an analytical model, as the Koopman operator framework provides a powerful means of data-driven linear representation. With the availability of a theoretical framework for analyzing system dynamics spectra and the dynamic mode decomposition (DMD) numerical algorithm for learning system dynamics from data~\cite{proctor2016dynamic,kutz2016dynamic}, this method avoids complex modeling procedures and facilitates efficient data-driven control design for highly nonlinear systems.

However, when applying DMD in data-driven modeling and using finite-dimensional representations for nonlinear dynamics, the mismatch between the approximate and actual models is inevitable. The mismatch comes from the availability and quality of data sets, sensor noises in measurements, and selections of basis functions. The plant-model mismatch inherently degrades the model-based control performance. In the worst case, if the Koopman model fails to include potentially excited modes, the control and observation spillovers will destabilize the closed-loop system~\cite{hagen2003spillover,balas1982toward}. The control inputs can excite the modes that are not in the space spanned by the chosen basis functions (control spillover). The measurements can contain modes that are not in the basis function space, and the control action will be corrupted (observation spillover). Therefore, the guaranteed robustness against model uncertainty and mismatch is a critical issue in the Koopman model-based control approach. 


It has been recognized that the noisy data can lead to biased models obtained by the DMD approach, because the DMD modes and magnitudes are perturbed by the noisy data. The bias caused by sensor noise was analyzed and quantified how bias depend on the size and level of the data in~\cite{dawson2016characterizing,hemati2017biasing}. De-biasing modifications to the DMD were further proposed to remove this bias under the assumptions of small noises. If noise property is known, then a direct correction can be applied. Otherwise, results from DMD forwards and backwards can be combined to remove biases. Another alternative is to apply total least-squares minimization algorithm to minimize the model bias. It was pointed out in~\cite{haseli2019approximating} that the noisy data through observable functions will not remain identically distributed. Therefore, an element-wise weighted total least squares (EWTLS) was proposed to find consistent approximation. However, the proposed noise-resilient method leads to a non-convex optimization. Ref.~\cite{wanner2022robust} investigated the stochastic Koopman operator for random dynamical systems where either the dynamics or observables are affected by noise, and a robust DMD algorithm was proposed to address the noise-induced bias.  

The above-mentioned studies are robust modeling techniques against noisy data. This paper provides a robust control perspective to address the noise-induced model bias. First, we will examine the bound of model bias under bounded noises, and present a pitfall case that the LQG controller based on biased model will fail to stabilize the nonlinear system. After that, we present a dual-loop control framework to bring in the robustness. The dual-loop robust control resembles an internal model control and consists of two loops for nominal control and robust control, respectively~\cite{zhou2001new,xiangchen_CDC,two_step,TIE_dualloop}. The nominal controller is designed based on the Koopman model, and it can be any stabilizing controller in the standard observer-based state-feedback form. The robust loop is placed with an $H_{\infty}$ controller. Thanks to the performance separation principle proved in~\cite{TIE_dualloop,xiangchen_CDC}, the nominal performance is independent of the robust controller matrices, and the robust performance is independent of the state-feedback and observer gains of the nominal controller~\cite{TIE_dualloop}. In other words,  the nominal control loop and robust control loop can be separately designed and coordinated between high performance and robustness against the model mismatch. The dual-loop control scheme can attenuate unmatched uncertainty/disturbance in $H_{\infty}$ manner to achieve a guaranteed robustness level. Different from the disturbance-observer-based control (DOBC)~\cite{DOBC_review}, the proposed dual-loop control does not estimate the model mismatch but relies on the $H_{\infty}$ control to improve the robustness.

The main contributions of the current work are summarized below:    
\begin{enumerate}
    \item Demonstrate the plant-model mismatch in the Koopman-identified model learned from the polluted data set and a pitfall of a Koopman-model-based optimal controller (LQG control) addressing nonlinear dynamics.
    \item Quantify the bound of the model mismatch from Koopman identification on the bounded noisy data. 
    \item Incorporate a dual-loop control scheme to the data-driven finite-dimensional Koopman model to compensate for the plant-model mismatch;  Give out the implementation of the data-driven Koopman identification and the dual-loop control to the nonlinear system; The dual-loop scheme is extended from our previous work~\cite{TIE_dualloop} in LTI systems to nonlinear systems. 
    \item Demonstrate the improvement in robustness by the dual-loop control applied on the Koopman model using an example of Van Der Pol oscillator. In high measurement noise levels, the nominal optimal controller designed from the Koopman model fails to regulate the nonlinear system, whereas the dual-loop scheme improves the robustness against the model mismatch.
\end{enumerate}
 
The rest of this paper is organized as follows. Section~\ref{section2} starts the discussion with the preliminaries of the Koopman operator theory for data-driven modeling. Section~\ref {section3} investigate the extended dynamic model decomposition (EDMD) to data with noise and evaluate the upper bounds of model mismatch. Section \ref{section4} presents the dual-loop robust control architecture with theoretical guarantees for system stability and robustness. Section~\ref{section5} presents an illustrative example of the Van Der Pol oscillator to demonstrate the advantages of the proposed control scheme, followed by conclusions in Section \ref{section6}. The following mathematical notations are used. $||\cdot||_{F}$ denotes the Frobenius norm, and $||T_{zw}||_{\infty}$ denotes the $H_{\infty}$ norm of the transfer function $T_{zw}\in RH_{\infty}$. ${\dagger}$ denotes the pseudo-inverse.

\section{Preliminaries}\label{section2}

\subsection{Koopman Operator Theory}
Consider a discrete-time nonlinear dynamical system 
\begin{equation}
    x_{t+1} = f(x_t).
\end{equation}
where $x_{t+1} \in \mathcal{X} \subset \mathrm{R}^n$ is the state at time step $t+1$ which evolves in time by non-linear state transition map $f: \mathcal{X} \rightarrow \mathcal{X}$. Given any {observable} $\psi \in \mathcal{F}$, the Koopman operator $\mathcal{K}: \mathcal{F} \rightarrow \mathcal{F}$ is defined as 
\begin{equation}
     \mathcal{K}\psi(x_t) = \psi(f(x_t)) = \psi(x_{t+1})
\end{equation}

Here, the Koopman operator $\mathcal{K}$ linearly evolves a nonlinear dynamics using observables $\psi$, which lift the finite-dimensional states $x_{t} \in \mathrm{R}^n$ to an infinite-dimensional observable space where original nonlinear dynamics can be represented by an infinite-dimensional linear system. 



Because the infinite-dimensional observables are practically infeasible to implement the theoretical transformations, an approximation-based approach is usually implemented to get a finite-dimensional approximation of the Koopman operator. This work uses extended dynamic mode decomposition (EDMD) \cite{williams2015data} with controls to get approximate transition matrices in observable space. 

\subsection{Extended Dynamic Mode Decomposition (EDMD)}
Consider the following discrete-time nonlinear dynamics,
\begin{equation}
\begin{aligned}
    x_{t+1}&= ~f(x_t,u_t), \\
     y_{t}&= ~g(x_t),
\end{aligned}
\label{eqn:NLdynamics}
\end{equation}
where $x_{t}$ is the state, $y_{t}$ is the output, $u_{t}$ is control inputs, and $f$, $g$ are the nonlinear functions. Using Koopman operator theory, the finite-dimensional Koopman model of the above nonlinear dynamics can be written as:
\begin{equation}
\begin{aligned}
    z_{t+1}&= A_pz_t + B_pu_t, \\
     y_{t}&= C_pz_t,
\end{aligned}
\label{eqn:NLdynamics}
\end{equation}
where, $z_{t} = \psi(x_t)$ is the lifted state in the observable space. The goal here is to estimate the unknown system transition matrices $A_p$, $B_p$, and $C_p$ using EDMD. Let us assume that the data snapshots of control trajectory $\mathbf{U}$, state trajectory $\mathbf{X}$ and output trajectory $\mathbf{Y}$ are available from the original nonlinear system \eqref{eqn:NLdynamics} with $N$ points. The entire data snapshot can then be written in matrix form as
\begin{equation}
\begin{array}{c}
     \textbf{X}^1 = [x_0,~x_1, ~\hdots~,~x_{N-1}], \quad \textbf{X}^2 = [x_1,~x_2, ~\hdots~,~x_{N}]\\
     \textbf{U} = [u_1,~u_2, ~\hdots~,~u_{N}], \quad 
     \textbf{Y} = [y_1,~y_2, ~\hdots~,~y_{N}]
\end{array}
\label{eqn:SplittingData}
\end{equation} 
where, where state data-matrix $\textbf{X}^2$ is shifted from the matrix $\textbf{X}^1$ by one-time forward step. The EDMD uses a set of observables to lift the nonlinear dynamics from their original state space to a high-dimensional observable space. Using the basis function vector $\psi(\cdot)$, the following high-dimensional matrix of lifted states can be obtained.
\begin{equation}
\begin{array}{c}
   \mathbf{X}_{\mathrm{lift}} = [\psi(x_0), ~\hdots~, \psi(x_{N})] \\  
\end{array}
    \label{eqn:LiftedMatrices}
\end{equation}
and the observables of $x_{t}$ is written as
\begin{equation}
\psi(x_{t}) = 
    \begin{bmatrix}
     \psi_1(x_{t}),
     \cdots,
     \psi_M(x_{t})
    \end{bmatrix}^{T}
\end{equation}
Here, $M$ is the dimension of the basis function vector, which is equivalent to the dimension of the lifted system in the observable space. 

Splitting the state matrix from \eqref{eqn:LiftedMatrices}, two data matrices $\textbf{X}^1_{\mathrm{lift}}$ and $\textbf{X}^2_{\mathrm{lift}}$ can be obtained as follows:
\begin{equation}
    \mathbf{X}_{\mathrm{lift}}^{1} = [\psi(x_0), ~\hdots~, \psi(x_{N-1})], \quad \mathbf{X}_{\mathrm{lift}}^{2} = [\psi(x_1), ~\hdots~, \psi(x_{N})].
    \label{eqn:LiftedMatrix_split}
\end{equation}

The control input matrix $\textbf{U}$ and output \textbf{Y} are not lifted to keep the input-output mapping. The approximate system matrices can be obtained by solving the least square regression problem \cite{korda2018linear}.

\begin{equation}
    \min\limits_{A_{p},B_p} ~ ||\mathbf{X}_{\mathrm{lift}}^{2} - A_{p}\mathbf{X}_{\mathrm{lift}}^{1} - B_{p}\textbf{U}||_F
    \label{eqn_LS_AB}
\end{equation}
\begin{equation}
    \min\limits_{C_{p}} ~ ||\textbf{Y} - C_{p}\mathbf{X}_{\mathrm{lift}}^{1}||_F
     \label{eqn_LS_CD}
\end{equation}

It is equivalently interpreted as linear system identification using input-state-output data.  The optimal solution of \eqref{eqn_LS_AB} will give the state-space matrices $A_{p}$ and $B_{p}$. Similarly, \eqref{eqn_LS_CD} will provide matrix $C_p$ for the output. The optimal solutions of least regression can be analytically expressed as:
\begin{equation}
    [A_{p},B_{p}] = \mathbf{X}_{\mathrm{lift}}^{2}\left[\mathbf{X}_{\mathrm{lift}}^{1};\mathbf{U}\right]^{\dagger}
\end{equation}
\begin{equation}
    C_{p} = \mathbf{Y}\left[\mathbf{X}_{\mathrm{lift}}^{1} \right]^{\dagger}
\end{equation}

Note that, in the absence of measurement noises, the model matrices $(A_{p}, B_{p}, C_{p})$ can render accurate approximations of input-output mapping by the high-dimensional linear system in observables space with small \textit{approximation error} and regression error. However, when measurement noises are present, the noise will perturb data snapshots, lifted/filtered by basis function $\psi(\cdot)$, and eventually render model mismatch with biased modes and magnitudes. The next section will investigate the model mismatch induced by noises, and quantify the bound of the model mismatch.

\section{EDMD on noisy data}\label{section3}
For the current work, we assume that the noisy measurement of state is $\hat{x}_{t} = x_{t} + n_{t}$ with small perturbations $n_{t}$, while the measurement errors of control inputs are not considered. The Koopman operator $\psi(\cdot)$ operates on the data of the perturbed states, rendering the perturbed lifted data-matrix $\hat{\mathbf{X}}_{\mathrm{lift}}$. Based on the corrupted measurement data, the expressions for the lifted state matrix can be written as: 
\begin{align}
    \hat{\mathbf{X}}^1_{\mathrm{lift}} = [\psi(x_0 + n_{0}), ~\hdots~, \psi(x_{N-1} + n_{N-1})] \\  \hat{\mathbf{X}}^2_{\mathrm{lift}} = [\psi(x_1 + n_{1}), ~\hdots~, \psi(x_{N} + n_{N})] \\
    \hat{\mathbf{Y}} = [y_1 + m_{1},~y_2+m_{2}, ~\hdots~,~y_{N} + m_{N}].
    \label{eqn:LiftedMatrix_split}
\end{align}

Thus, the perturbed least-square regression for Koopman model identification produces $(A, B_{2}, C_{2})$, which will be used for controller design. Note that, $(A, B_{2})$ are the perturbed transition matrices from the nominal matrices $(A_{p}, B_{p}, C_{2})$. The least-square problem formulation can be written as

\begin{equation}
    \min\limits_{{A},{B}_2} ~ ||\hat{\mathbf{X}}^2_{\mathrm{lift}} ~-~ {A}\hat{\mathbf{X}}^1_{\mathrm{lift}} ~-~ {B}_2\textbf{U}||_F
    \label{pertub_Koopman_AB}
\end{equation}
\begin{equation}
    \min\limits_{{C}_{2}} ~ ||\hat{\mathbf{Y}} ~-~ {C}_{2}\mathbf{X}^1_{\mathrm{lift}}||_F
    \label{pertub_Koopman_CD}
\end{equation}

Assuming the selected basis function is continuously differentiable (for instance, polynomial, Gaussian, radial basis function kernel), then, using Taylor expansion, the observables with perturbed data can be approximated by its first-order equation as

\begin{equation}
    \psi(x_{k}+n_{k}) \approx \psi(x_{k}) + \frac{\partial \psi}{ \partial x_{k}} n_{k}
\end{equation}
The overall lifted perturbed data snapshot can be approximated as 

\begin{align}
    \hat{\mathbf{X}}^1_{\mathrm{lift}} = & [\psi(x_0)+\frac{\partial \psi}{ \partial x_{0}}n_{0}, ~\hdots~, \psi(x_{N-1})+\frac{\partial \psi}{ \partial x_{N-1}}n_{N-1}] \\  
    = & \mathbf{X}^1_{\mathrm{lift}} + [\frac{\partial \psi}{ \partial x_{0}}n_{0}, ~\hdots~, \frac{\partial \psi}{ \partial x_{N-1}}n_{N-1}] = \mathbf{X}^1_{\mathrm{lift}} + E_{1}\\
    \hat{\mathbf{X}}^2_{\mathrm{lift}} = & [\psi(x_1)+\frac{\partial \psi}{ \partial x_{1}}n_{1}, ~\hdots~, \psi(x_{N})+\frac{\partial \psi}{ \partial x_{N}}n_{N}] \\ 
      = & \mathbf{X}^2_{\mathrm{lift}} + [\frac{\partial \psi}{ \partial x_{1}}n_{1}, ~\hdots~, \frac{\partial \psi}{ \partial x_{N}}n_{N}] = \mathbf{X}^2_{\mathrm{lift}} + E_{2}
    \label{eqn:LiftedMatrix_split}
\end{align}

Therefore, the nominal model identification can be derived in terms of perturbed data measurements as 
\begin{equation}\label{pertub_Koopman}
    [{A}_p,{B}_p] = {\mathbf{X}}^2_{\mathrm{lift}}\left[{\mathbf{X}}^1_{\mathrm{lift}}; \mathbf{U}\right]^{\dagger} = (\hat{\mathbf{X}}^2_{\mathrm{lift}}-E_{2})\left[(\hat{\mathbf{X}}^1_{\mathrm{lift}}-E_{1}); \mathbf{U}\right]^{\dagger}
\end{equation}
where, $E_1$ and $E_2$ are the additional terms generated due to due to noise being lifted to observable space.
We are interested in exploring how the noisy data filtered by basis function will perturb the regression matrices $({A}, {B}_{2})$ from nominal matrices $(A_{p}, B_{p})$. The model mismatch is denoted as $(\Delta A, \Delta B) = (A, B_{2}) - (A_{p},B_{p}) $. Now, we will derive the bound of $(\Delta A, \Delta B)$. The spectrum (natural) norm is of interest and defined as $\|A \|_{2} = \max_{\|x\|_{2}\neq 0} \frac{\|Ax\|_{2}}{\|x \|_{2}}$, which interprets the model mismatch in terms of $l_{2}$ norm of system signals.

\begin{theorem}\cite{stewart1977perturbation}
    For any $S$, $T$ and perturbation $E$ such that $T = S + E$, the perturbed pseudo-inverse $T^{\dagger} - S^{\dagger}  = \mathcal{O}(E)$ is bounded as
    \begin{equation}
        \| T^{\dagger} - S^{\dagger} \|  = \| \mathcal{O}(E) \|  \leq \mu \max\left\{  \|S^{\dagger} \|_{2}^{2},  \|T^{\dagger} \|_{2}^{2} \right\} \|E\|.
    \end{equation}
    $\mu = \frac{1+\sqrt{5}}{2}$ if the spectrum norm is considered. 
\end{theorem}

Let $S = \left[(\hat{\mathbf{X}}^1_{\mathrm{lift}}-E_{1}); \mathbf{U}\right]$, $T = \left[\hat{\mathbf{X}}^1_{\mathrm{lift}}; \mathbf{U}\right]$, so $E = \left[E_{1}; 0\right]$.  Therefore, from equation (\ref{pertub_Koopman}), the perturbed Koopman model matrices can be derived as 
\begin{align}
[{A}_p,{B}_p] = & (\hat{X}_{\mathrm{lift}}^{2} - E_2)S^{\dagger}\\
= & (\hat{\mathbf{X}}^2_{\mathrm{lift}}-E_{2}) (T^{\dagger} - \mathcal{O}(E)) \\ 
= & \hat{\mathbf{X}}^2_{\mathrm{lift}}T^{\dagger} - E_{2}T^{\dagger} - (\hat{\mathbf{X}}^2_{\mathrm{lift}} - E_{2})\mathcal{O}(E) \\ 
= & [{A}, {B}_{2}] - E_{2}T^{\dagger} - \hat{\mathbf{X}}^2_{\mathrm{lift}}\mathcal{O}(E) +E_{2}\mathcal{O}(E)\\
= & [{A}, {B}_{2}] - \mathbf{\Delta}_{A,B}
\end{align}
where
\begin{align}
    \mathbf{\Delta}_{A,B} = [\Delta A, \Delta B] = E_{2}T^{\dagger} + \hat{\mathbf{X}}^2_{\mathrm{lift}}\mathcal{O}(E) -E_{2}\mathcal{O}(E)
\end{align}
The perturbations on system matrices include first-order perturbation $E_{2}T^{\dagger}, \mathbf{X}^2_{\mathrm{lift}}\mathcal{O}(E)$, and second-order perturbation $E_{2}\mathcal{O}(E)$. The maximum disturbance caused by the measurement noise to the system identification process can be obtained by finding out the upper bound of the $\mathbf{\Delta}_{A,B}$. The spectrum norm of system matrices perturbation is bounded and calculated by
\begin{equation}
    \| [\Delta {A}, \Delta{B}]  \| \leq  \|E_{2}\| \|T^{\dagger}\| + \| \hat{\mathbf{X}}^2_{\mathrm{lift}} - E_2\| \| \mathcal{O}(E_{1}) \| = U
\end{equation}



Now, for the output transition matrix, the perturbation from $C_{p}$ to ${C}_{2}$ is evaluated. $\hat{\textbf{Y}}$ is a noisy output data matrix by perturbation $m_{t}$ as $\hat{y}_t = y_t + m_t$
\begin{align}
    \hat{\textbf{Y}} = ~[y_1 + m_1, y_2 + m_2, ~\hdots~ , y_N + m_N] = \textbf{Y} + M
\end{align}

Writing the expression for nominal output matrix $C_p$ in terms of perturbed lifted data matrix as
\begin{align}\label{output_eqn}
    C_p & = \mathbf{Y}\left[ \mathbf{X}^1_{\mathrm{lift}} \right]^{\dagger} \\
    {C}_{p} & = (\hat{\mathbf{Y}} - M)\left[\hat{\mathbf{X}}^1_{\mathrm{lift}} - E_1 \right]^{\dagger} 
\end{align}
Assuming $[\hat{\mathbf{X}}^1_{\mathrm{lift}} - E_1]^{\dagger} = ~[\hat{\mathbf{X}}^1_{\mathrm{lift}}]^{\dagger} - \mathcal{O}(E_{1})$, where $\mathcal{O}(E_{1})$ is the perturbation on pseudo-inverse. Therefore, 
\begin{align}\label{output_solved}
    {C}_{p} = & ~(\hat{\mathbf{Y}} - M)([\hat{\mathbf{X}}^1_{\mathrm{lift}}]^{\dagger} - \mathcal{O}(E_{1}))\\
    = & ~ \hat{\mathbf{Y}}\left[\hat{\mathbf{X}}^1_{\mathrm{lift}} \right]^{\dagger}  - M[\hat{\mathbf{X}}^1_{\mathrm{lift}}]^{\dagger} - \hat{Y}\mathcal{O}(E_{1}) + M\mathcal{O}(E_{1})\\
    = & ~ C_{2} - \Delta C
\end{align}    
where,
\begin{align}
    \Delta C = M[\hat{\mathbf{X}}^1_{\mathrm{lift}}]^{\dagger} + \hat{Y}\mathcal{O}(E_{1}) - M\mathcal{O}(E_{1})
\end{align}

The perturbation on the system matrix $C$ has first-order $M[\hat{\mathbf{X}}^1_{\mathrm{lift}}]^{\dagger}, ~ Y\mathcal{O}(E_{1})$ and second-order $M\mathcal{O}(E_{1})$.  The spectral norm of the perturbed output matrix $\Delta C = C_{2} - C_{p}$ can be bounded as:
\begin{equation}
    \|\Delta {C}\| \leq  \|M\| \|[\hat{\textbf{X}}^1_{lift}]^{\dagger}\| + \| \hat{Y} - M\| \|\mathcal{O}(E_{1}) \| = V
\end{equation}

Based on the above calculation for perturbation on Koopman model identification due to measurement noises, the model mismatch is bounded. Denote the model mismatch as $f_{s}(\Tilde{x}_{t}, u_{t}, t)= A \Tilde{x}_{t} + B_{2} u_{t} - {A}_{p}\Tilde{x}_{t} - {B}_{p} u_{t} = [\Delta A, \Delta B][\Tilde{x}_{t}^{T},  
u_{t}^{T}]^{T}$, so $\|f_{s} \|_{2} \leq  \|U [\Tilde{x}_{t}^{T}
u_{t}^{T}]^{T}\|_{2}$. This provides a sector bound of the model mismatch of $f_{s}$. Similarly, denote $v_{s}(\Tilde{x}_{t}, t) = C_{2} \Tilde{x}_{t} - {C}_{p} \Tilde{x}_{t} = \Delta C \Tilde{x}_{t}$, so $\|v_{s} \|_{2} \leq V \| \Tilde{x}_{t}\|_{2}$. Therefore, the true Koopman model can be expressed by the perturbed Koopman model with bounded model mismatch as follows
\begin{equation}\label{Koopman_wNoiseBound}
	\begin{aligned} 
	\Tilde{x}_{t+1} & = A \Tilde{x}_{t} + B_{2} u_{t} + B_{1}w_{t} - f_{s} \\ 
 z_{t} &=C_{1} \Tilde{x}_{t}+D_{12} u_{t} \\ 
 y_{t} &= C_{2} \Tilde{x}_{t} - v_{s} \end{aligned}
\end{equation}
where $C_{2}$ is an identity matrix since we are considering a state feedback system, $C_{1}, D_{12}$ are user-selected weighting matrices for the performance output $z_{t}$. $w_{t}$ is external disturbance, and $f_{s}$ can be written as
\begin{align}
   \| f_{s}(\Tilde{x}_t, u_t, t)\|_{2} \leq \| \left[\begin{array}{cc}
        U_1 & U_2 \\
    \end{array}\right] \left[\begin{array}{c}
         \Tilde{x}_{t}  \\
         u_{t}
    \end{array}\right] \|_2, \ \& \ v_{s}(\Tilde{x}_t, u_t, t) \leq \| \left[\begin{array}{cc}
        V_1 & 0 \\
    \end{array}\right] \left[\begin{array}{c}
         \Tilde{x}_{t}  \\
         u_{t}
    \end{array}\right] \|_2
\end{align}
where 
\begin{align}
    \| \left[ U_1, U_2\right] \|_{2} = U \ \& \ \| \left[ V_1, 0\right] \|_{2} = V
\end{align}


The above discussion shows that the Koopman model is subject to biases when identified from noisy data of states. A similar conclusion can be drawn with noisy input and output data. It also provides the upper bound on the model mismatch for states. The next section will begin discussing the development of a robust dual-loop control strategy for the biased Koopman model. Robust dual-loop control takes into account model mismatch and generates compensation control inputs to achieve guaranteed robustness. In other words, the Koopman model-based control is improved in its ability to reject measurement noises.


\section{Dual-loop robust control}\label{section4}

Consider the discrete-time Koopman model matrices from perturbed data points~\eqref{Koopman_wNoiseBound}. It is assumed that $(A, B_{2})$ are stabilizable, and $(C_{2}, A)$ are detectable, such that there exist stabilizing controllers for the system. The $H_{\infty}$ control usually produces conservative performance since it aims to handle the worst-case model uncertainty. However, the dual-loop control can relieve the conservativeness by introducing the compensation block \cite{TIE_dualloop}.

\begin{figure}[!htb]
	\centering
	\includegraphics[width=0.65\linewidth]{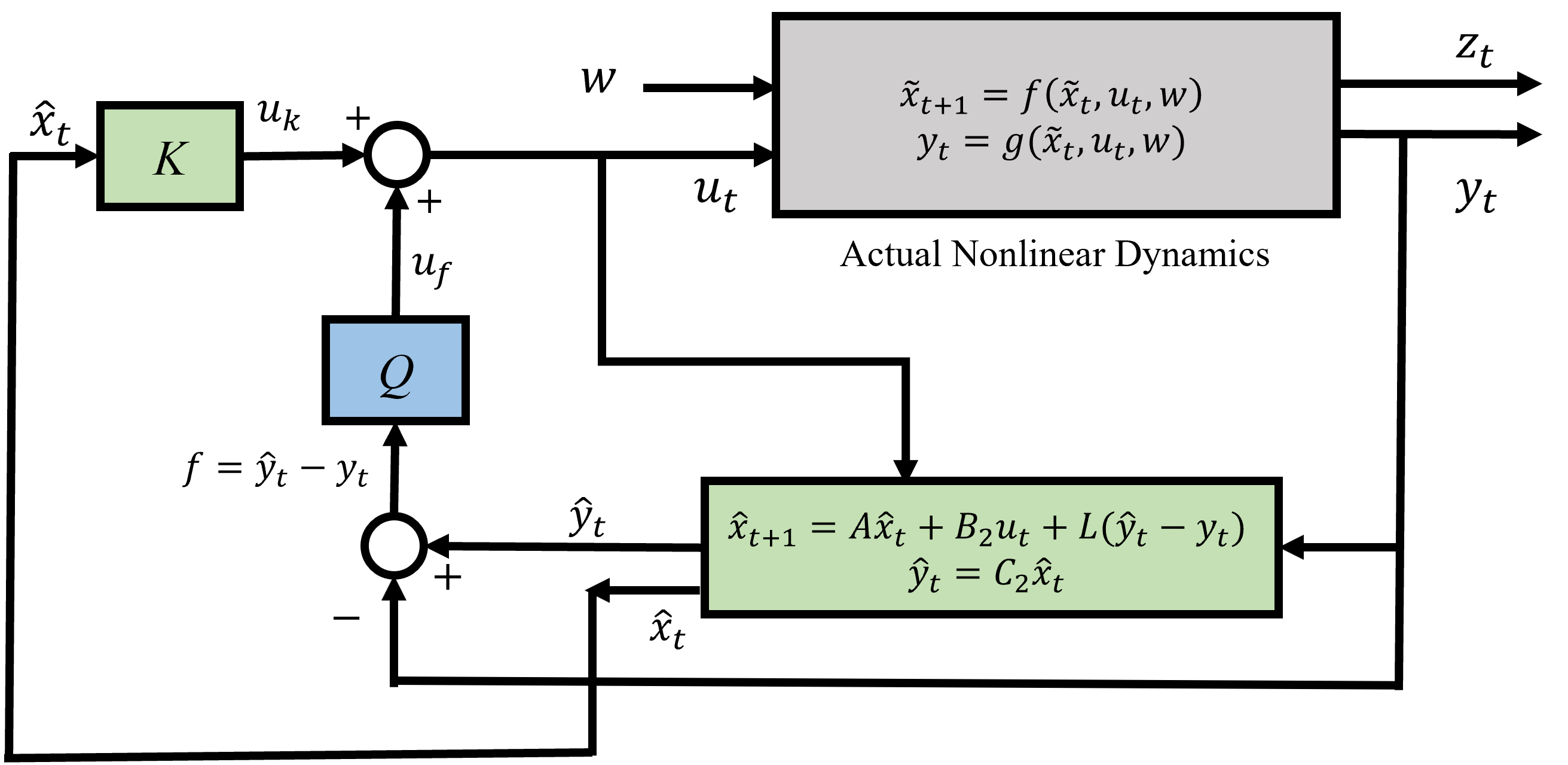}
	\caption{The dual-loop control strategy with the nominal control (green) and robust loop (blue) designed on the data-driven Koopman model.}
	\label{fig:Implementation}
\end{figure}

As shown in Fig.~\ref{fig:Implementation}, the dual-loop control scheme has two control loops for nominal performance and robustness, respectively. The first loop consists of a Luenberger observer and state-feedback control with gains $(K, L)$ based on the Koopman model. The second loop is designed with the operator $Q$ in the $H_{\infty}$ approach to bring in the guaranteed robust stability and performance. The residual signal $f_{t} = \hat{y}_{t} - y_{t}$ is the deviation of estimated outputs from the Koopman observer and the actual outputs, indicating the mismatch between the Koopman model and the actual system. When there is no model mismatch, the residual signal $f_{t} = 0$ goes through $Q$ and produces $u_{f} = Q f_{t} = 0$. The overall controller still remains the nominal controller. If there is a model mismatch, the state observer renders biased estimations and outputs. The robust control loop $Q$ is then triggered, and the extra control signal $u_{f,t} = Q f_{t}$ compensates for the nominal control input $u_{k,t}$. Moreover, $u_{f,t}$ is generated by $Q$, filtering the residual signal and adapting to the 'size' of model mismatch to bring robustness to the overall system.

The design of the dual-loop control scheme for the system~\eqref{Koopman_wNoiseBound} is formulated as follows: Let $(K, L)$ be any given nominal controller stabilizing Koopman (nominal) model identified from noisy data, and take the form
\begin{equation}\label{nominalcontrol}
	\begin{aligned} {\hat x}_{t+1}&=A \hat x_{t}+B_{2} u_{t} +L(\hat y_{t}-y_{t}) \\ \hat y_{t} &=C_{2} \hat{x}_{t}, \;\ u_{k,t}=K\hat{x}_{t}. \end{aligned}
\end{equation}
Assuming $Q$ to be the dynamic form:
\begin{equation}\label{Qcontrol}
	\begin{aligned} {x}^{Q}_{t+1}&=A_{Q} x^Q_{t}+B_Q f_{t} \\ u_{f,t} &=C_Q x_{t}^Q + D_{Q}f_{t}, \end{aligned}
\end{equation}
The control design problem is to find controller matrices $A_Q$, $B_Q$, $C_{Q}$ and $D_Q$ to achieve robust stability under sector-bounded model mismatch $f_{s}, v_{s}$
 and achieve robust performance $\left\|T_{zw}\right\|_{\infty} < \gamma$, where $T_{zw}$ is the closed-loop system from $w$ to $z$.

\subsection{Nominal Controller Design as LQG}
The nominal controller is designed by the LQG controller to obtain optimal performance. With nominal controller gains $(K, L)$, the identified system \eqref{Koopman_wNoiseBound} can be rewritten into \eqref{aug_model_raw} by including estimation error $e_{t} = \Tilde{x}_{t}- \hat{x}_{t}$ into state vector.

\begin{equation}\label{aug_model_raw}
	\begin{aligned} \Tilde{x}_{t+1} &=\left(A+B_{2} K\right) \Tilde{x}_{t} -B_{2} K e_{t} + B_{2} u_{f,t} - f_s + B_{1}w_{t}\\ {e}_{t+1} &=\left(A+L C_{2}\right) e_{t}+ B_{1}w_{t} - f_s - Lv_s\\ z_{t} &=\left(C_{1}+D_{12} K\right) \Tilde{x}_{t}-D_{12} K e_{t}+D_{12} u_{f,t} \\ f_{t} &=-C_{2} e_{t} + v_s \end{aligned}
\end{equation}

It can be seen from \eqref{aug_model_raw} that the estimation error $e_{t}$ and state $\Tilde{x}_{t}$ are perturbed by plant-model mismatch $f_s$ and $v_s$, and the residual $f_{t}$ consists of estimation error and model mismatch $v_{s}$. Thus, the compensation input $u_{f,t}$ is generated by the operator $Q $ on residual signal $f_{t}$ to bring in the guaranteed robustness.



By defining
\begin{equation*}
	\bar{x}_{t}=\left[\begin{array}{c} \tilde{x}_{t} \\ 
 {e}_{t}\end{array}\right], \quad \bar{A}=\left[\begin{array}{cc}{A+B_{2} K} & {-B_{2} K} \\ {0} & {A+L C_{2}}\end{array}\right], \quad \bar{D}_{12}=D_{12}, \bar{D}_{21} = [\mathbf{0} \quad I]
\end{equation*}
\begin{equation*}
	\bar{N}_{t}=\left[\begin{array}{c} f_s \\ 
 v_s\end{array}\right], \quad \bar{F}=\left[\begin{array}{cc} -I & 0\\ -I & -L\end{array}\right], \quad \bar{B}_{1} =  \left[\begin{array}{c}{B_{1}} \\ {B_{1}}\end{array}\right], \quad \bar{B}_{2}=\left[\begin{array}{c}{B_{2}} \\ \mathbf{0}\end{array}\right],
\end{equation*}
\begin{equation*}
	\bar{C}_{1}=\left[\begin{array}{cc}{C_{1}+D_{12} K} & {-D_{12} K}\end{array}\right], \quad \bar{C}_{2}=\left[\begin{array}{cc} \mathbf{0} & {-C_{2}}\end{array}\right]
\end{equation*}

the augmented system can be rearranged as
\begin{equation}\label{aug_model}
	\begin{aligned} \bar{x}_{t+1} &= \bar{A} \bar{x}_{t}+\bar{F}\bar{N}_{t} + \bar{B}_{1} w_{t} + \bar{B}_{2} u_{f,t} \\ z &=\bar{C}_{1} \bar{x}_{t}+ \bar{D}_{12} u_{f,t} \\ f_{t} &=\bar{C}_{2} \bar{x}_{t} + \bar{D}_{21}\bar{N}_{t}. \end{aligned}
\end{equation}

where, 
\begin{align}
   \| \bar{N}_{t}(\bar{x}_{t}, u_{f,t}, t)\|_{2} = \| \left[\begin{array}{c} f_s(\bar{x}_{t}, u_{f,t}, t) \\ 
 v_s(\bar{x}_{t}, u_{f,t}, t)\end{array}\right]\|_{2} \leq \| \left[\begin{array}{cc}
        U' & U_2 \\
        V' & 0
    \end{array}\right] \left[\begin{array}{c}
         \bar{x}_{t}  \\
         u_{f,t}
    \end{array}\right] \|_2 \\
    U' = \left[ \begin{array}{cc}
        U_1 + U_2K & -U_2K \\
    \end{array}\right], \quad V' = \left[ \begin{array}{cc}
        V_1 & 0 \\
    \end{array}\right]
\end{align}

Now, the new 'output' $f_{t}$ is the deviation between estimated and actual outputs, and $u_{f,t}$ provides an additional degree of freedom of control. For ease of expression, system~\eqref{Koopman_wNoiseBound} is referred to as \textit{original} system, and system~\eqref{aug_model} is referred to as \textit{augmented} system.

\subsection{$H_{\infty}$ Control Design Under Bounded Model Mismatch}
With the definition of $Q$ defined in \eqref{Qcontrol}, the closed-loop system can be written as:

\begin{equation}\label{aug_model_new}
	\begin{aligned} 
 \bar{x}_{t+1} &= (\bar{A} + \bar{B}_2D_Q\bar{C}_2)\bar{x}_{t}+ \bar{B}_{2}C_Qx^Q_t + \bar{F}\bar{N}_{t} + \bar{B}_{1}w_{t} - B_2D_Qv_s\\
 x^Q_{t+1} &= A_Qx^Q_t + B_Q(\bar{C}_2\bar{x}_t + \bar{D}_{21}\bar{N}_{t})\\
 z_t &= \bar{C}_{1} \bar{x}_{t}+ \bar{D}_{12}C_Qx^Q_t + \bar{D}_{12}D_Q(\bar{C}_2\bar{x}_t -v_s) \\
 \end{aligned}
\end{equation}

Considering $D_Q = 0$, \eqref{aug_model_new} becomes
\begin{equation}\label{aug_model_final}
	\begin{aligned} 
 \bar{x}_{t+1} &= \bar{A}\bar{x}_{t} + \bar{B}_{2}C_Qx^Q_t + \bar{F}\bar{N}_{t} + \bar{B}_{1}w_{t} \\
 x^Q_{t+1} &= A_Qx^Q_t + B_Q(\bar{C}_2\bar{x}_t + \bar{D}_{21}\bar{N}_{t})\\
 z_t &= \bar{C}_{1} \bar{x}_{t}+ \bar{D}_{12}C_Qx^Q_t \\
 \end{aligned}
\end{equation}

Defining 
\begin{equation*}
	\bar{x}_{cl,t}=\left[\begin{array}{c} \bar{x}_{t} \\ x^Q_{t}\end{array}\right], \quad  \bar{F}_{cl} = \left[\begin{array}{c}
	     \bar{F} \\
      B_{Q}\bar{D}_{21}
	\end{array} \right],\quad \bar{A}_{cl}=\left[\begin{array}{cc} \bar{A} & \bar{B}_{2}C_Q \\ B_Q\bar{C}_2 & A_Q \end{array}\right],
\end{equation*}
\begin{equation*}
    \bar{B}_{cl} = \left[\begin{array}{c} \bar{B}_1 \\0 \end{array}\right], \quad \bar{C}_{cl}= \left[\begin{array}{c c}\bar{C}_{1} & \bar{D}_{12}C_Q\end{array}\right]
\end{equation*}

Final closed-loop equations for the entire system will be
\begin{equation}\label{Final_model}
\begin{aligned}
    \bar{x}_{cl,t+1} &= \bar{A}_{cl}\bar{x}_{cl,t} + \bar{B}_{cl}w_{t} + \bar{F}_{cl}{\bar{N}}_{t} \\
    z_{t} &= \bar{C}_{cl}\bar{x}_{cl,t}
\end{aligned}  
\end{equation}

Here, $\bar{F}_{cl}{\bar{N}}_{t}$ is the nonlinear term associated with noise perturbation to Koopman model, and the upper bound of ${\bar{N}}_{t}(\bar{x}_{cl,t})$ can be written as:
\begin{align}
    \| \bar{N}_{t}(\bar{x}_{cl,t}, t) \|_{2} = \|\left[\begin{array}{c} f_s(\bar{x}_{cl,t}, t) \\ 
 v_s(\bar{x}_{cl,t}, t)\end{array}\right]\|_{2} \leq \| \left[\begin{array}{cc}
        U' & UC_Q \\
        V' & 0
    \end{array}\right] \|_2
\end{align}

Equation (\ref{Final_model}) represents the final closed-loop model with feedback from dual-loop control. The condition to achieve the asymptotic stability and robust performance guarantee is provided in the lemma \ref{Lema1} below.

\begin{lemma}\label{Lema1}\cite{wei2012composite}
Consider a discrete-time system 
\begin{equation}\label{DT_system}
    \begin{aligned}
        x_{t+1} &= Gx_t + Ff(x_t,t) + H w_t\\
        z_t &= Cx_t \\
    \end{aligned}
\end{equation}
where the nonlinearity $f(x_{t},t)$ satisfies the sector bound denoted by a constant matrix $U_{*}$, $ \| f(x_{t}) \| \leq  \|U_{*} x_{t} \|$. For given parameters $\lambda >0$, $\gamma > 0$, if there exists $P>0$ satisfying
\begin{equation}\label{LMI_DT_RL}
   \left[ \begin{array}{cccc}
        G^{T}PG - P + \lambda^2U^T_*U_* & G^{T}PF & G^{T}PH & C^T \\
        \\
        F^{T}PG & F^TPF -\lambda^2I & F^TPH & 0 \\
        \\
        H^{T}PG & H^TPF & H^TPH - \gamma^2I & 0 \\
        \\
        C & 0 & 0 & -I \\
    \end{array} \right] < 0
\end{equation}

then \eqref{DT_system} is robustly asymptotically stable with sector-bounded $f(x_{t},t)$ in the absence of $w_{t}$, and satisfies robust performance $\|z\|_2 < \gamma \|w\|_2$ in the presence of $w_{t}$. 
\end{lemma}

\begin{proof}
The entire proof is divided into two parts. Part 1 provides the necessary conditions for asymptotic stability, while Part 2 extends the proof to guarantee robust control performance.

$\Rightarrow$ \textbf{Part 1 - Robust stability under sector-bounded $f(x)$}

Consider a Lyapunov function $V(t)$ and $V(t+1)$ at time step $t$ and $t+1$, respectively, as
    \begin{equation}
        V(t) = x^{T}_{t}Px_{t} + \lambda^{2}\sum_{\tau = 0}^{k-1} (\|U_{*}x_\tau\|^{2} - \| f(x_\tau)\|^{2} )
    \end{equation}
    \begin{equation}
    V(t+1) = x^{T}_{t+1}Px_{t+1} + \lambda^{2}\sum_{\tau = 0}^{k} (\|U_{*}x_\tau\|^{2} - \| f(x_\tau)\|^{2} )
    \end{equation}

    Here, $V(t) \geq 0 $ since $\| f(x_{t}) \| \leq  \|U_{*}x_{t} \|$. In the absence of $w_{t}$, $\Delta V(t)$ can be derived as 
    \begin{equation}
        \begin{array}{rl}
            \Delta V(t)  & = V(t+1) - V(t)  \\
            & =  x^{T}_{t+1}Px_{t+1} - x^{T}_{t}Px_{t}  + \lambda^{2}(\|U_{*} x_{t}\|^{2} - \| f(x_{t})\|^{2} ) \\ 
            & = [Gx_{t} + Ff(x_{t})]^{T}P[Gx_{t} + Ff(x_{t})] \\
            & - x^{T}_{t}Px_{t} + \lambda^{2}[x^{T}_{t}U_{*}^{T}U_{*}x_{t} - f^{T}(x_{t})f(x_{t}) ] \\ 
            & = \left[\begin{array}{c}
                 x_{t} \\
                  f
            \end{array}\right]^{T} \left[\begin{array}{cc}
               G^{T}PG - P + \lambda^{2}U^T_*U_* & G^{T}PF \\
                F^{T}PG & F^{T}PF -\lambda^{2}I
            \end{array} \right] \left[\begin{array}{c}
                 x \\
                  f
            \end{array}\right]
        \end{array}
    \end{equation}

    If there exists a positive definite matrix $P$ that satisfies $$Q_{1} = \left[\begin{array}{cc}
               G^{T}PG - P + \lambda^{2}U^T_*U_* & G^{T}PF \\
                F^{T}PG & F^TPF -\lambda^{2}I
            \end{array} \right] < 0, $$ 
            
            then $\Delta V(t) = V(t+1) - V(t) < 0 $, $\lim_{k\rightarrow \infty}{V(t)} = 0$, therefore, the system achieves asymptotic stability with sector-bounded $f(x_{t})$. 

$\Rightarrow$ \textbf{Part 2 - Robust performance}:

            Now, considering a storage function 
            \begin{equation}
                J = \sum_{\tau = 0}^{\infty}{J(\tau)} = \sum_{\tau = 0}^{\infty}\left( \|z_\tau\|^{2} - \gamma^{2}\|w_\tau\|^{2} + \Delta V(\tau) \right)
            \end{equation}
            Then, in the presence of $w_{t}$, it is easy to derive the expression for $J(t)$ as
            \begin{equation}
                \begin{array}{rl}
                     J(t)
                     & = \Delta V(t) + \left[x^{T}_{t}C^{T}Cx_{t} - \gamma^{2} w^{T}_{t}w_{t}\right] \\
                     & = [Gx_{t} + Ff(x_{t}) + H w_{t}]^{T}P[Gx_{t} + Ff(x_{t}) + Hw_{t}] \\
                     & - x^{T}_{t}Px_{t}  + \lambda^{2}[x^{T}_{t}U_{*}^{T}U_{*}x_{t} - f^{T}_{t}f_{t} ] + [x^{T}_{t}C^{T}Cx_{t} - \gamma^{2} w^{T}_{t}w_{t}] \\
                     & = \left[\begin{array}{c}
                 x_{t} \\
                  f \\
                  w_{t}
            \end{array}\right]^{T}  Q_{2}\left[\begin{array}{c}
                 x_{t} \\
                  f \\
                  w_{t}
            \end{array}\right] 
                \end{array}
            \end{equation}

            where $Q_{2} = \left[ \begin{array}{ccc}
        G^{T}PG - P + \lambda^{2}U^T_*U_* + C^{T}C & G^{T}PF & G^{T}PH  \\
        F^{T}PG & F^TPF-\lambda^{2}I & F^TPH        \\
        H^{T}PG & H^TPF & H^TPH - \gamma^2I
    \end{array} \right]$, 
    
If there exists a positive definite matrix $P$ that satisfy $Q_{2} < 0$, which is equivalent to \eqref{LMI_DT_RL}, then $ J(t) < 0$ at any time instant $k$. Under zero initial condition, $J(0) = V(0) = 0$, and $J \geq \sum_{\tau = 0}^{\infty}\left( \|z_\tau\|^{2} - \gamma^{2}\|w_\tau\|^{2}\right)$,  therefore, $\sum_{\tau = 0}^{\infty}\left( \|z(\tau)\|^{2} - \gamma^{2}\|w(\tau)\|^{2}\right) < 0$. The robust performance $\|z\|_2 < \gamma \|w\|_2$ is satisfied with sector-bounded nonlinearity and external disturbance $w_{t}$.
\end{proof}

Based on the above Lemma, the proposed theorem below provides the controller synthesis conditions for the $H_{\infty}$ controller under sector-bounded model mismatch.

\begin{theorem}\label{Theorem1}[controller synthesis]
    For any given parameter $\lambda$, $\gamma$, and sector bound $U_*$, if there exist positive definite matrices $X_{1}$, $Y_{1}$, and matrices $\hat{A}_{Q}, \hat{B}_{Q}, \hat{C}_{Q}$ such that

\begin{equation}
\left[\begin{array}{cc}
    X_{1} &  I \\
     I & Y_{1}
\end{array} \right] > 0
\end{equation}


\begin{equation}\small
\left[\begin{array}{ccccccccc}
-X_1 & I & \bar{A}X_1 + \bar{B}_2\hat{C}_Q & \bar{A} &  \bar{F} & \bar{B}_{1}  & 0 & 0 & 0\\
* & -Y_1 & \hat{A}_{Q} & Y_1\bar{A} + \hat{B}_Q\bar{C}_2 & Y_1\bar{F} + \hat{B}_Q\bar{D}_{21} & Y_{1}\bar{B}_{1}  & 0 & 0 & 0\\
* & * & -X_1 & I & 0 & 0 &  X_1\bar{C}_1^T + \hat{C}_Q^T\bar{D}_{12}^T  & X_1U'^{T} + \hat{C}_Q^TU_2^T & X_1V'^T \\
* & * & * & -Y_1 & 0 & 0 & \bar{C}_{1}^T &  U'^{T} & V'^T\\
* & * & * & * & -\lambda^{2} & 0 & 0 & 0 & 0\\
* & * & * & * & * & -\gamma^2 I & 0 & 0 & 0\\
* & * & * & * & * & * & -I & 0 & 0\\
* & * & * & * & * & * & * & -\frac{1}{\lambda^{2}}I & 0\\
* & * & * & * & * & * & * & * & -\frac{1}{\lambda^{2}}I 
\end{array}\right] < 0
\end{equation}

Moreover, the controller is given as :
\begin{equation}
    \left[\begin{array}{cc}
        A_Q & B_Q \\
        C_Q & 0 \\
    \end{array}\right] = \left[\begin{array}{cc}
        Y_2 & Y_1\bar{B}_2 \\
        0 & I
    \end{array} \right]^{-1}\left[\left[\begin{array}{cc}
        \hat{A}_Q & \hat{B}_Q \\
         \hat{C}_Q & 0
    \end{array} \right] - \left[\begin{array}{cc}
        Y_1\bar{A}X_1 & 0 \\
         0 & 0
    \end{array} \right]\right]\left[\begin{array}{cc}
        X^T_2 & 0 \\
         \bar{C}_2X_1 & I
    \end{array} \right]^{-1}
\end{equation}
\end{theorem}

\vspace{10pt}
\begin{proof}

Upon comparing equation (\ref{Final_model}) and (\ref{LMI_DT_RL}), one can confirm the following equality: $[\bar{A}_{cl}, \bar{F}_{cl}, \bar{B}_{cl}, \bar{C}_{cl}] = [G, F, H, C]$. Re-writing the expression for $Q_2$ in terms of the original closed-loop system matrices will achieve

\begin{equation}
    Q_{2} = \left[ \begin{array}{ccc}
        \bar{A}_{cl}^{T}P\bar{A}_{cl} - P + \lambda^{2}U^T_*U_* + \bar{C}_{cl}^{T}\bar{C}_{cl} & \bar{A}_{cl}^{T}P\bar{F}_{cl} & \bar{A}_{cl}^{T}P\bar{B}_{cl}  \\
        \bar{F}_{cl}^{T}P\bar{A}_{cl} & \bar{F}_{cl}^TP\bar{F}_{cl}-\lambda^{2}I & \bar{F}_{cl}^TP\bar{B}_{cl}        \\
        \bar{B}_{cl}^{T}P\bar{A}_{cl} & \bar{B}_{cl}^TP\bar{F}_{cl} & \bar{B}_{cl}^TP\bar{B}_{cl} - \gamma^2I
    \end{array} \right] 
\end{equation}

where, $U_*$ is the sector bound given as $[U' \ U_2C_Q]$. Using Schur complement, $Q_{2}$ can be converted to an equivalent $Q_{3}$ given as

\begin{equation}
    Q_3 = 
   \left[ \begin{array}{cccccc}
        - P & P\bar{A}_{cl} & P\bar{F}_{cl} & P\bar{B}_{cl} & 0 & 0\\
        * & -P & 0 & 0 & \bar{C}_{cl}^{T} & U^{T}_{*}\\
        * & * & - \lambda^2 I & 0 & 0 & 0\\
        * & * & * & -\gamma^{2} I & 0 & 0\\
        * & * & * & * & - I & 0 \\
        * & * & * & * & 0 & -\frac{1}{\lambda^{2}} I \\
    \end{array} \right] < 0 
\end{equation}

Let $X_{cl}^{T} = \left[\begin{array}{cc}
    X_{1} &  X_{2} \\
     I & 0
\end{array} \right]$, $Y_{cl} = \left[\begin{array}{cc}
    I &  0 \\
     Y_{1} & Y_{2}
\end{array} \right]$, $Y_1X_1 + Y_2X^T_2 = I$, and $X_{cl}^{T}P = Y_{cl}$.

Therefore, $X_{cl}^{T}PX_{cl} = \left[\begin{array}{cc}
    X_{1} &  I \\
     I & Y_{1}
\end{array} \right].$ Left multiply $T = \mathrm{diag}(X_{cl}^{T},X_{cl}^{T},I,I,I,I)$ and right multiply $T^{T}$ on $Q_{3}$, one obtains

\begin{equation}
     Q_4 = 
   \left[ \begin{array}{cccccc}
        - X_{cl}^{T}PX_{cl} & Y_{cl}\bar{A}_{cl}X_{cl} & Y_{cl}\bar{F}_{cl} & Y_{cl}\bar{B}_{cl} & 0 & 0\\
        * & -X_{cl}^{T}PX_{cl} & 0 & 0 & X_{cl}^{T}\bar{C}_{cl}^{T} & X_{cl}^{T}U^{T}_{*}\\
        * & * & - \lambda^2 I & 0 & 0 & 0\\
        * & * & * & -\gamma^{2} I & 0 & 0\\
        * & * & * & * & - I & 0 \\
        * & * & * & * & 0 & -\frac{1}{\lambda^{2}} I \\
    \end{array} \right]
\end{equation}

Using the expressions for $\bar{A}_{cl}, \bar{F}_{cl}, \bar{B}_{cl}, \text{and} \bar{C}_{cl}$, the following expressions can be easily derived.

\begin{align}\label{2ndTerm}
    Y_{cl}\bar{A}_{cl}X_{cl} = \left[ \begin{array}{cc}
        \bar{A}X_1 + \bar{B}_2C_QX_2^T & \bar{A} \\
         Y_1\bar{A}X_1 + Y_1\bar{B}_2C_QX_2^T + Y_2B_Q\bar{C}_2X_1 + Y_2A_QX^T_2 & Y_1\bar{A} + Y_2B_Q\bar{C}_2 
    \end{array}\right]
\end{align}
\begin{align}\label{3rdTerm}
    Y_{cl}\bar{F}_{cl} = \left[ \begin{array}{c}
         \bar{F}  \\
         Y_1\bar{F}
    \end{array}\right], \ Y_{cl}\bar{B}_{cl} = \left[ \begin{array}{c}
         \bar{B}_1  \\
         Y_1\bar{B}_1
    \end{array}\right], 
\end{align}
\begin{align}\label{OtherTerm}
    X_{cl}^{T}\bar{C}_{cl}^{T} = \left[ \begin{array}{c}
        X_1\bar{C}_1^T + \hat{C}_Q^T\bar{D}_{12}^T \\
        \bar{C}_1^T 
    \end{array}\right], X_{cl}^{T}U_*^{T} = \left[ \begin{array}{cc}
        X_1U'^T + X_2C^T_QU_2^T & X_1V'^T \\
        U'^T & V'^T 
    \end{array}\right]
\end{align}

Defining change of variable as below.
\begin{equation}\label{ChangeofVar}
    \left[\begin{array}{cc}
        \hat{A}_Q & \hat{B}_Q \\
         \hat{C}_Q & 0
    \end{array} \right] = \left[\begin{array}{cc}
        Y_2 & Y_1\bar{B}_2 \\
        0 & I
    \end{array} \right]\left[\begin{array}{cc}
        A_Q & B_Q \\
        C_Q & 0 \\
    \end{array}\right]\left[\begin{array}{cc}
        X^T_2 & 0 \\
         \bar{C}_2X_1 & I
    \end{array} \right] + \left[\begin{array}{cc}
        Y_1\bar{A}X_1 & 0 \\
         0 & 0
    \end{array} \right], 
\end{equation}

Simplifying the right-hand side of the equation (\ref{ChangeofVar}) will give
\begin{equation}\label{NewControlMatrix}
    \left[\begin{array}{cc}
        \hat{A}_Q & \hat{B}_Q \\
         \hat{C}_Q & 0
    \end{array} \right] = \left[\begin{array}{cc}
        Y_1\bar{A}X_1 + Y_1\bar{B}_2C_QX_2^T + Y_2B_Q\bar{C}_2X_1 + Y_2A_QX^T_2 & Y_2{B}_Q \\
        C_QX_2^T & 0
    \end{array} \right]
\end{equation}

Upon substituting the new variables in equations (\ref{2ndTerm}) and (\ref{3rdTerm}), the expression for $Y_{cl}\bar{A}_{cl}X_{cl}$ can be written as
\begin{align}\label{New2ndterm}
    Y_{cl}\bar{A}_{cl}X_{cl} = \left[ \begin{array}{cc}
        \bar{A}X_1 + \bar{B}_2\hat{C}_Q & \bar{A} \\
         \hat{A}_Q & Y_1\bar{A} + \hat{B}_Q\bar{C}_2
    \end{array}\right]
\end{align}

Substituting the expressions \ref{New2ndterm}, \ref{3rdTerm}, and \ref{OtherTerm} into the expression for $Q_4$ will give the desired \textit{LMI} given in Theorem 2.

\end{proof}

\section{Simulation Results}\label{section5}

The current work considers a nonlinear Van Der Pol oscillator for demonstrating the theoretical formulations. The non-conservative system dynamics with nonlinear damping is given as 
\begin{equation}
\begin{aligned}
    \Dot{x}_1 &= x_2 \\ 
    \Dot{x}_2 &= \mu(1 - x^2_1)x_2 - x_1 + u
\end{aligned}
\end{equation}

The coefficient $\mu$ determines the damping coefficient, and for the current work, it was chosen as $\mu = 1$. The goal is to stabilize the origin by designing an output feedback control law based on data. For the current problem, the output $y$ of the system is assumed to be both the state measurements. Therefore, the final dynamics of the system can be written as

\begin{equation}
    \Dot{X} = f(X,u), \quad 
    y = X
\end{equation}
where, \begin{equation*}
    f = \left[\begin{array}{c}
         x_2\\
         \mu(1 - x^2_1)x_2 - x_1 + u  
    \end{array}\right]
\end{equation*}

The dynamics is first discretizing using Runge-Kutta 4 with a step size of $\Delta t = 0.01$. The discretised dynamics is then simulated for 20 seconds to collect data points. For lifting the dynamics, monomials of the order of up to 5 were chosen, making the observable space 21-dimensional. For data generation, a uniformly distributed random initial condition $x_0 \in [-1, 1]$ and control input $u \in [-10,10]$ was chosen. A uniform Gaussian white noise is added to the output measurements to simulate noisy observation. Based on the identified Koopman model, the proposed closed-loop control based on theorem (\ref{Theorem1}) is designed. 


Before moving forward, it is important to see how adding a small noise significantly changes the Koopman-identified model and why it is necessary to consider model mismatch due to measurement noise. Figure \ref{fig:NoiseIdentification} shows two cases of implementing the Koopman identification approach using measurements with different noise levels. Figure \ref{KI_SmallNoise} on the left is with a smaller noise variance $\sigma = 0.01$, and Figure \ref{KI_BigNoise} with bigger noise variance $\sigma = 0.05$. One obvious observation is that in both cases, the identified dynamics do not match true Van Der Pol oscillator behaviour. However, with higher noise case, the identified Koopman dynamics converges faster to the equilibrium point, which is the origin. Designing any feedback control strategies based on the Koopman-identified model and then implementing it on an actual system could lead to unwanted control behaviour as the model mismatch between the true and identified model increases drastically with measurement corruption.

\begin{figure*}[h!]\centering
\subfloat[\scriptsize $\sigma_{noise}$ - 0.01]{\minipage{0.52\textwidth}
 \includegraphics[width=\linewidth]{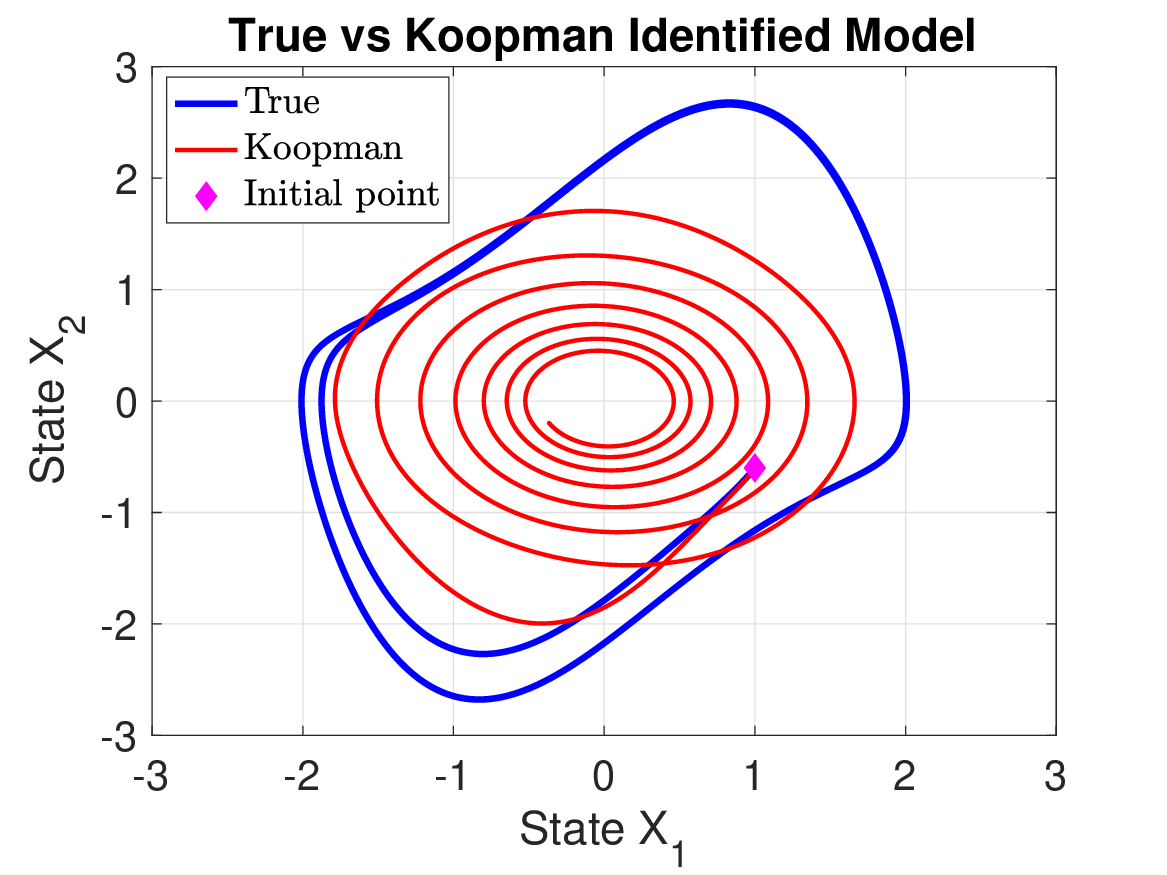}\label{KI_SmallNoise}
\endminipage}
\subfloat[\scriptsize $\sigma_{noise}$ - 0.05]{\minipage{0.52\textwidth}
\includegraphics[width=\linewidth]{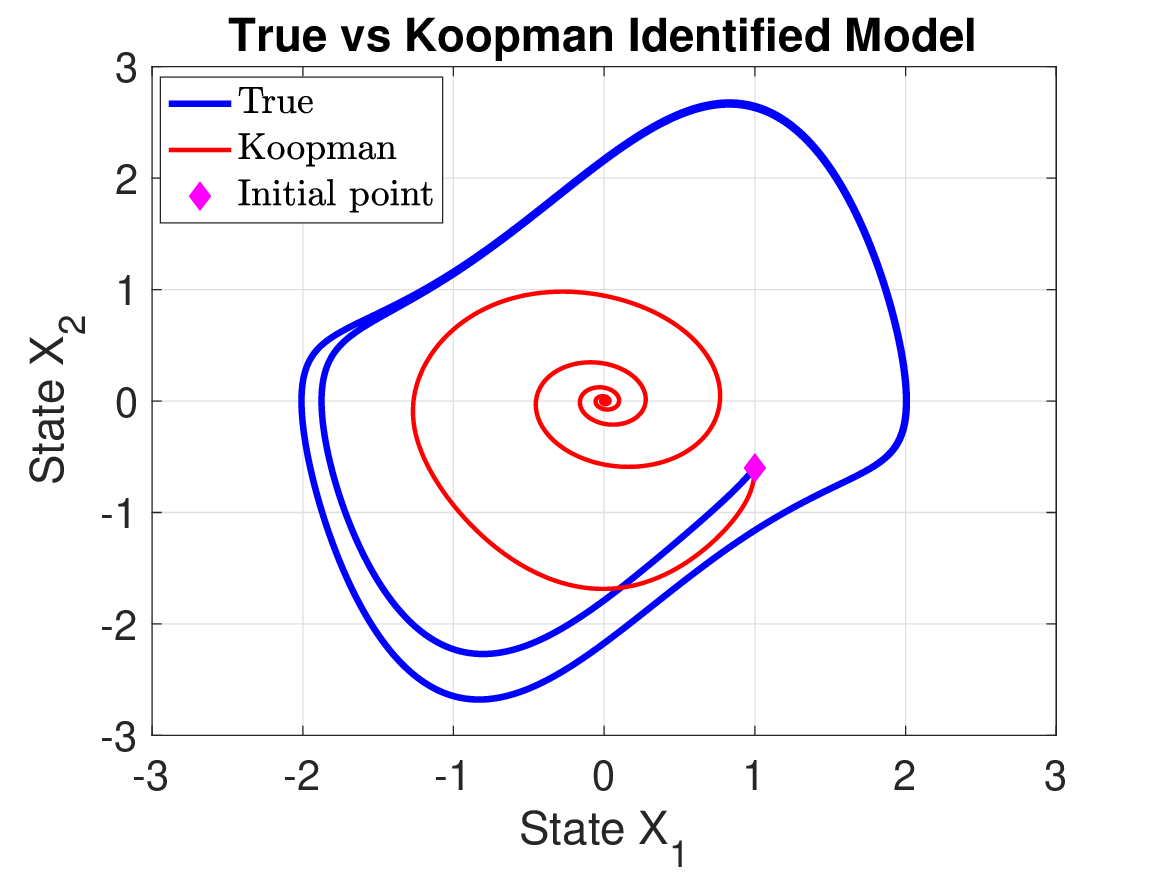}\label{KI_BigNoise}
\endminipage}
\caption{Changes in Koopman identified model with different noise standard deviation}
\label{fig:NoiseIdentification}
\end{figure*}

To showcase the efficacy of our proposed approach to handle the model mismatch in the Koopman-identified model due to measurement noise, we are working with the case when noise variance is $\sigma = 0.01$. The goal is to regulate the dynamics to achieve equilibrium at the origin. Linear Quadratic Regulator (LQG) control is also designed using the identified linear Koopman model to perform the comparison and show how the noise affects convergence. Both the control strategies designed using the Koopman-identified model are implemented to the original nonlinear dynamics. Figure \ref{fig:LQG_DualLoop_n01} shows the closed-loop control performance using both LQG and the proposed dual-loop. Because the identified model is an approximation of the nonlinear dynamics, the LQG control cannot accurately estimate the system states and, therefore, cannot converge to the origin. This is not the case with the $dual-loop$ control strategy. 



\begin{figure*}[h!]\centering
\subfloat[\scriptsize State $X_1$ - LQG]{\minipage{0.33\textwidth}
 \includegraphics[width=\linewidth]{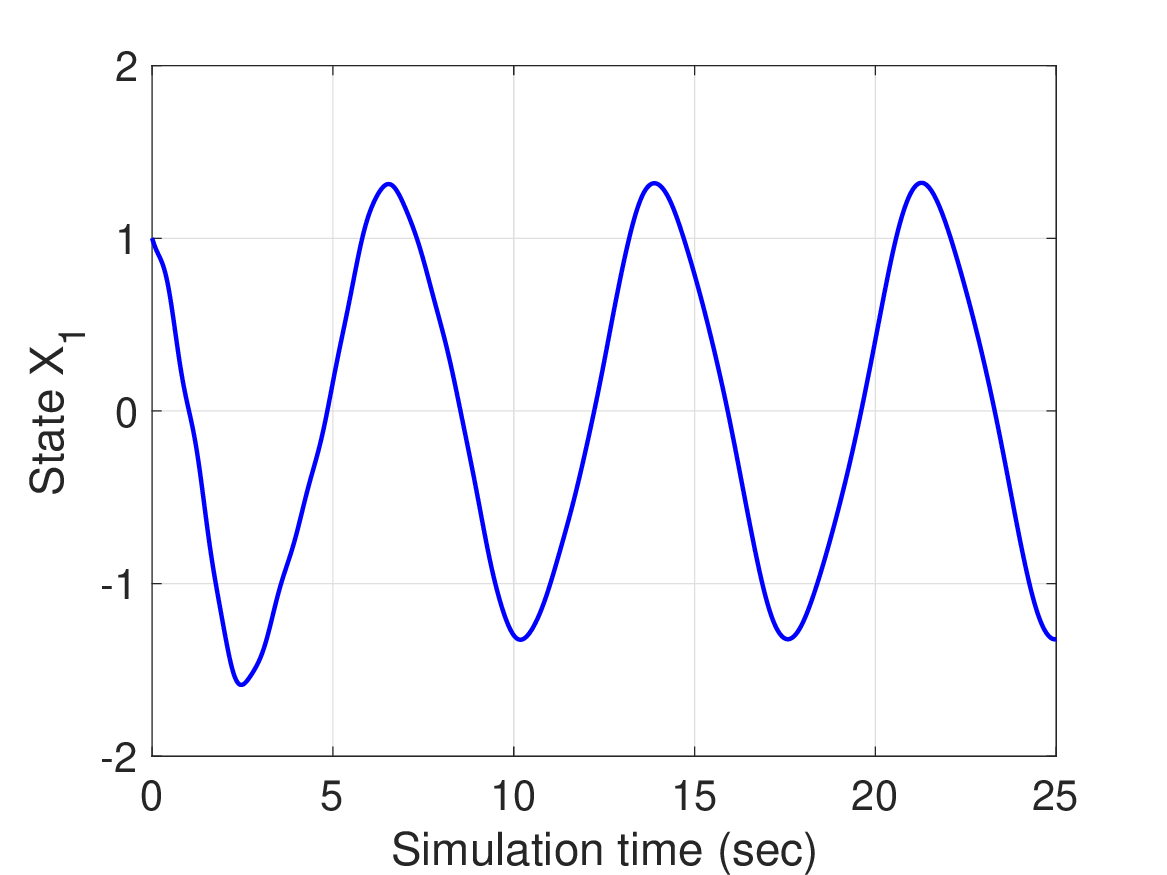}\label{figure2a}
\endminipage}
\subfloat[\scriptsize State $X_2$ - LQG]{\minipage{0.33\textwidth}
\includegraphics[width=\linewidth]{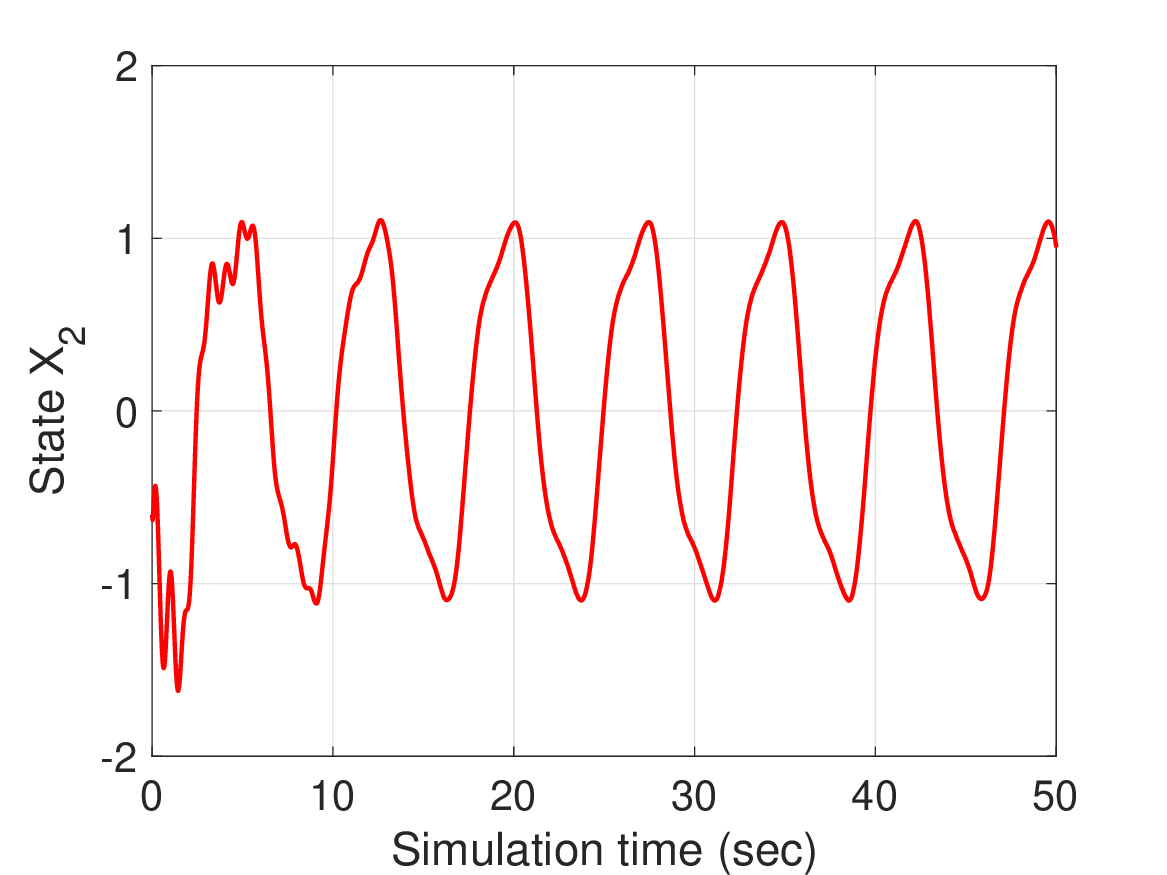}\label{figure2b}
\endminipage}
\subfloat[\scriptsize Control Input - LQG]{\minipage{0.33\textwidth}
\includegraphics[width=\linewidth]{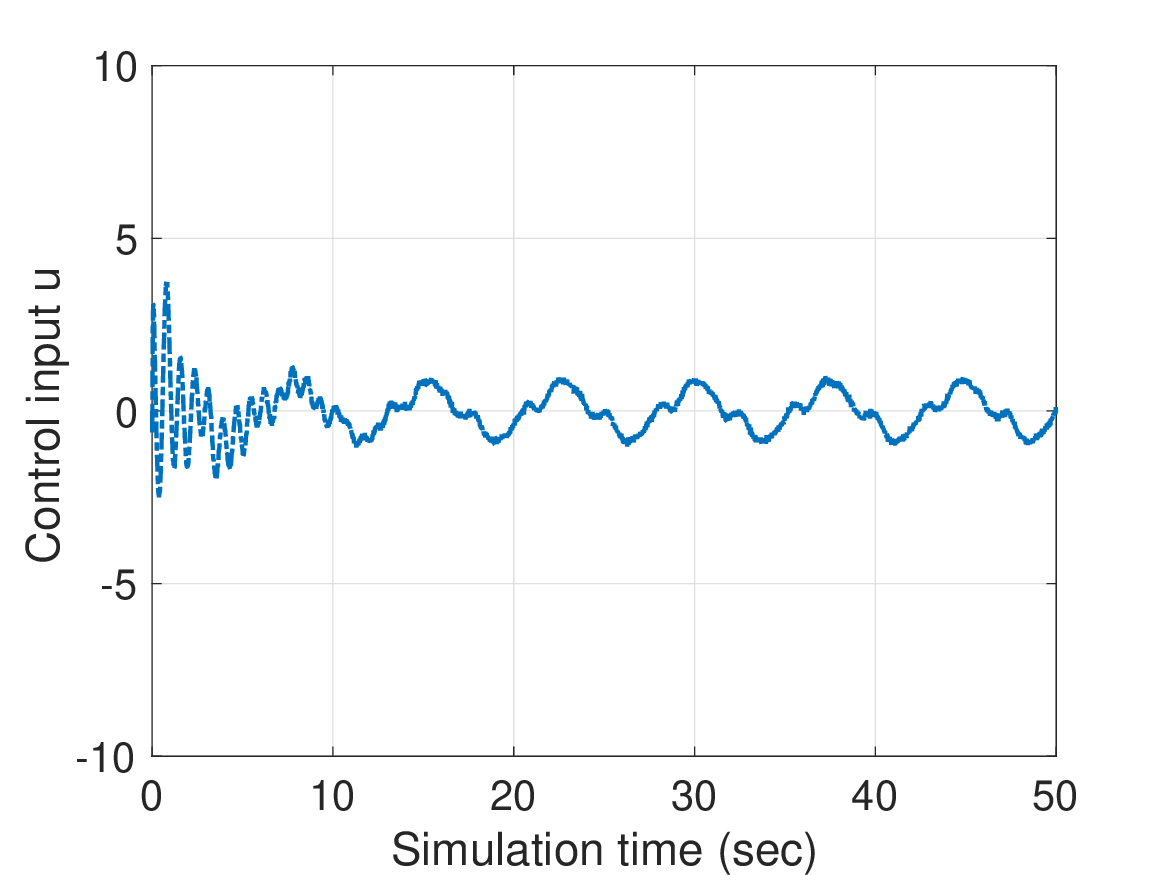}\label{figure2c}
\endminipage}\\
\subfloat[\scriptsize State $X_1$ - Dual loop]{\minipage{0.33\textwidth}
 \includegraphics[width=\linewidth]{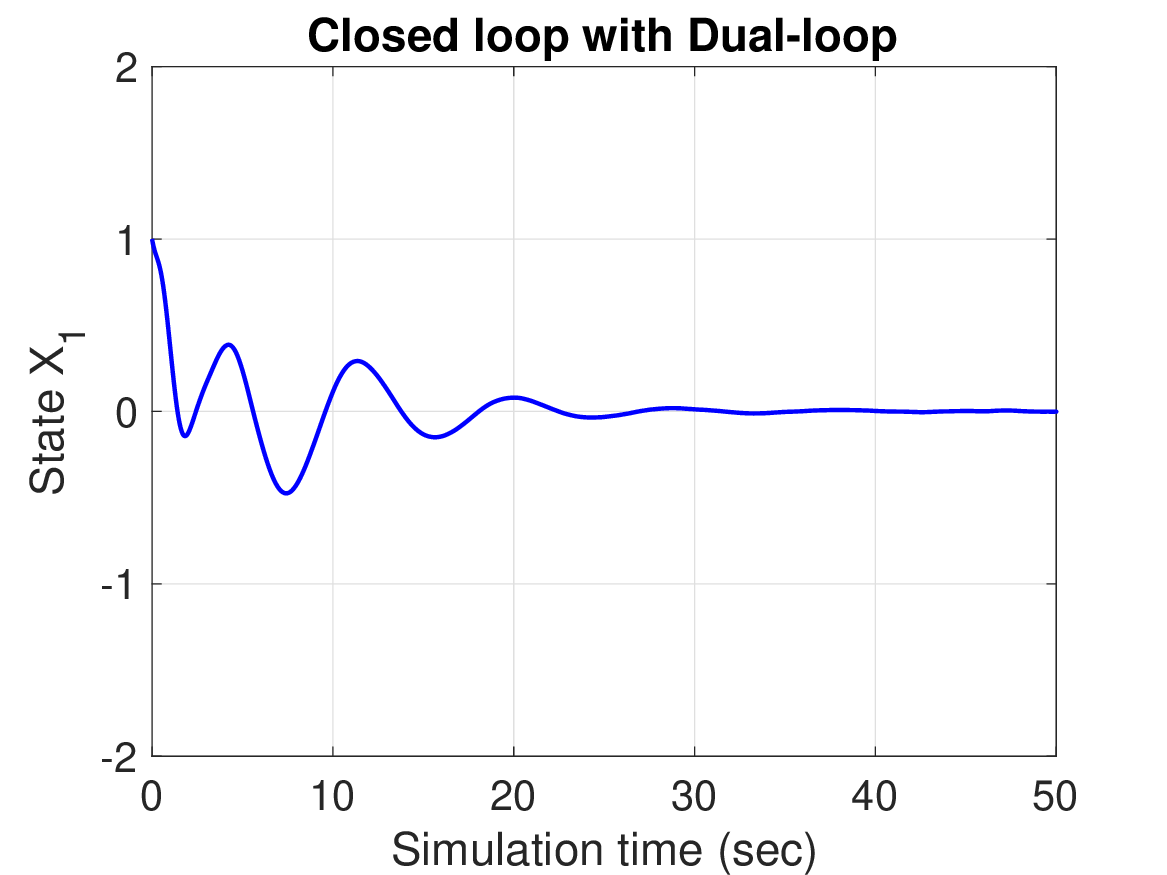}\label{figure2d}
\endminipage}
\subfloat[\scriptsize State $X_2$ - Dual loop]{\minipage{0.33\textwidth}
\includegraphics[width=\linewidth]{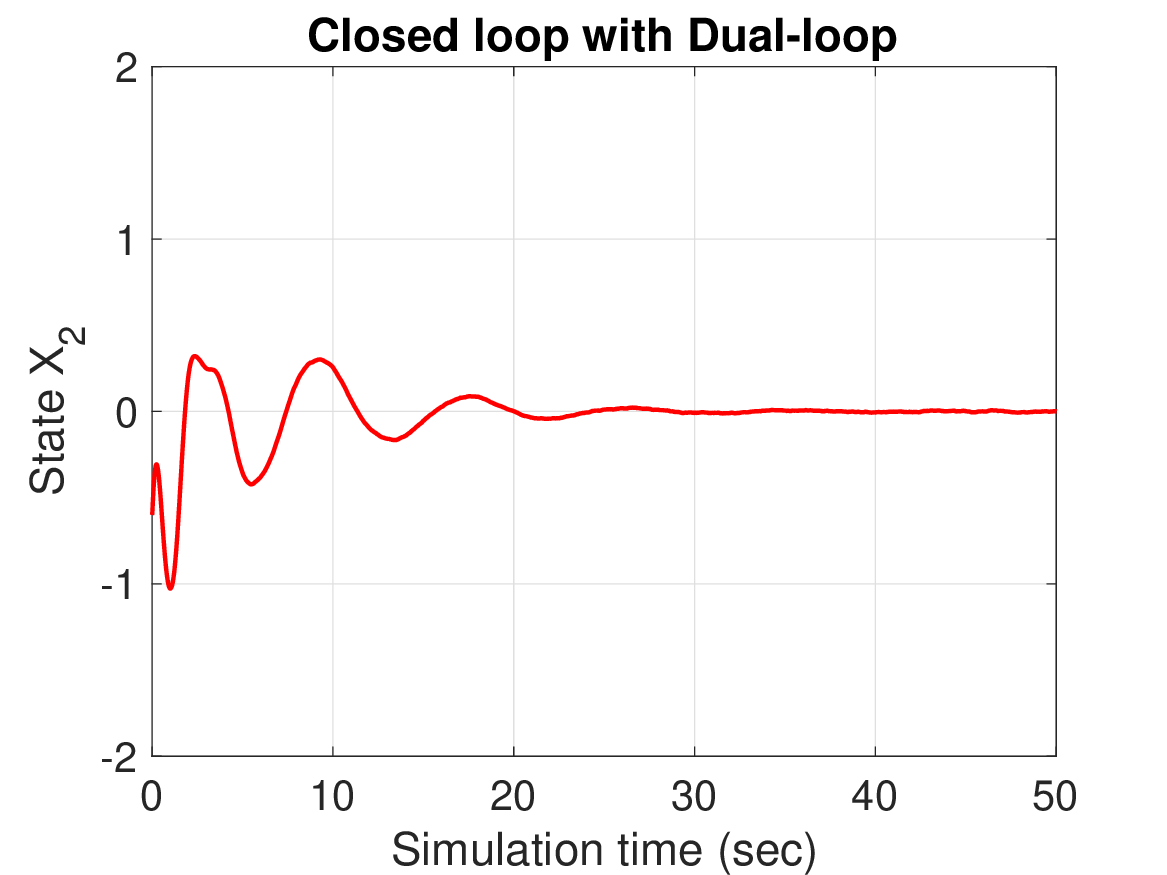}\label{figure2e}
\endminipage}
\subfloat[\scriptsize Control Input - Dual loop]{\minipage{0.33\textwidth}
\includegraphics[width=\linewidth]{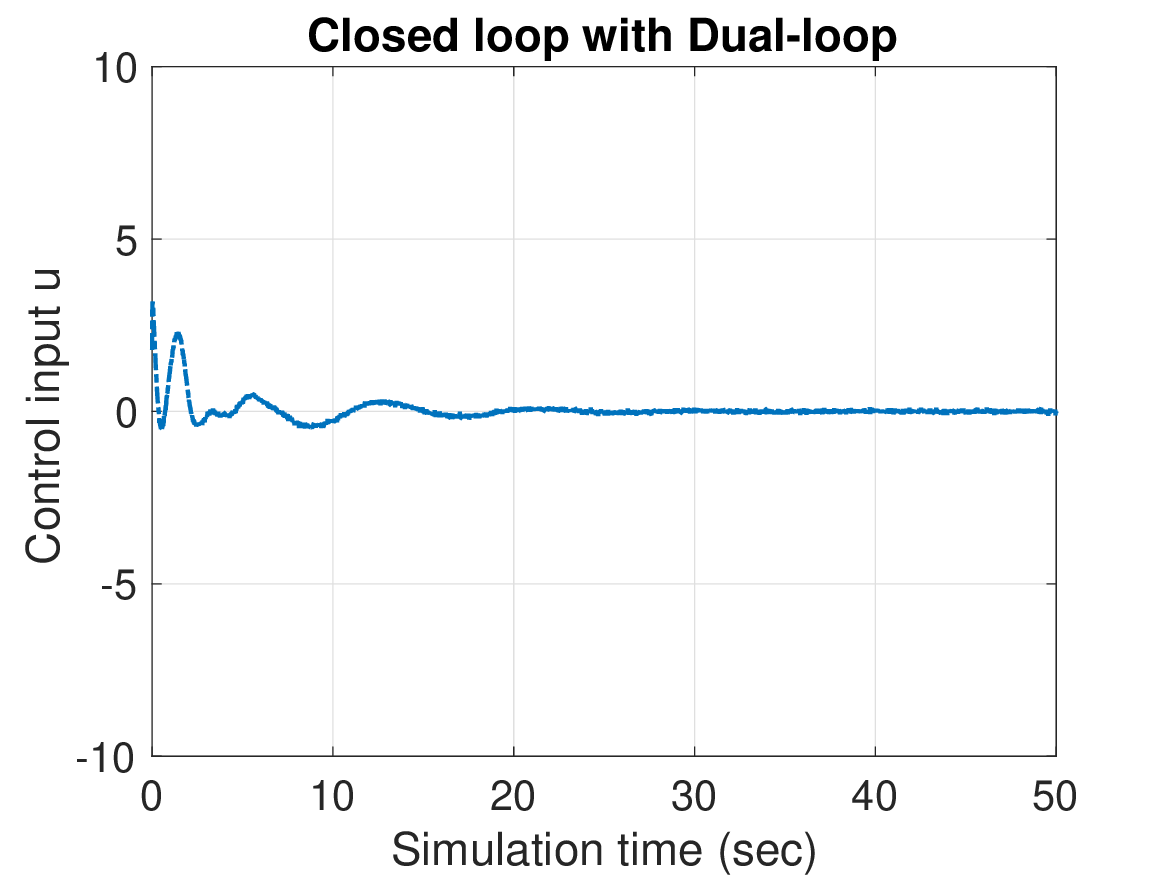}\label{figure2f}
\endminipage}
\caption{Closed-loop results with LQG control and dual-loop control at \textit{noise standard deviation} = 0.01}
\label{fig:LQG_DualLoop_n01}
\end{figure*}

Figure \ref{fig:DeltaY_DualLoop_n01} also shows how the output error $f = \hat{y} - y$ changes with the dual loop over time. Clearly, the dual loop manages to push the output error to zero, minimizing the model mismatch and leading to the desired convergence of the overall dynamics.

\begin{figure*}[h!]\centering
\subfloat[\scriptsize State $X_1$ - LQG]{\minipage{0.52\textwidth}
 \includegraphics[width=\linewidth]{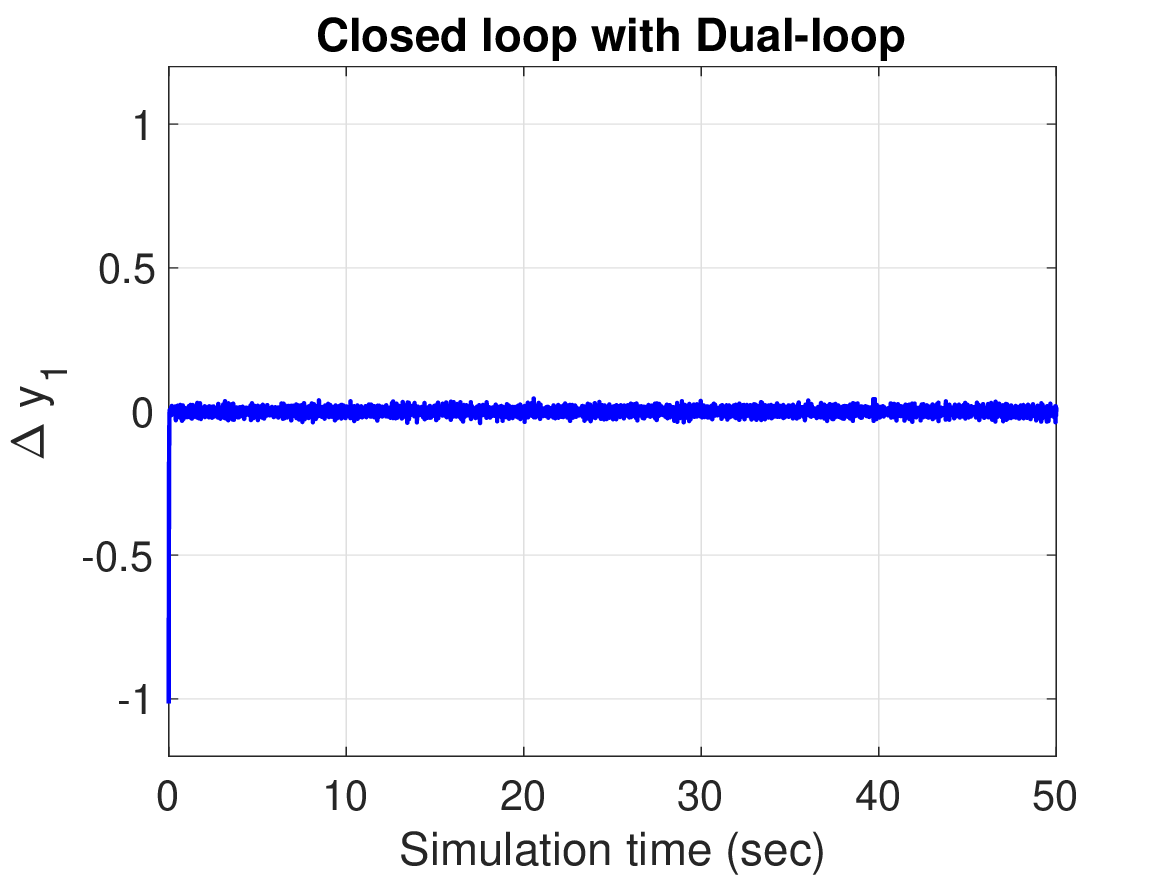}\label{figure2a}
\endminipage}
\subfloat[\scriptsize State $X_2$ - LQG]{\minipage{0.52\textwidth}
\includegraphics[width=\linewidth]{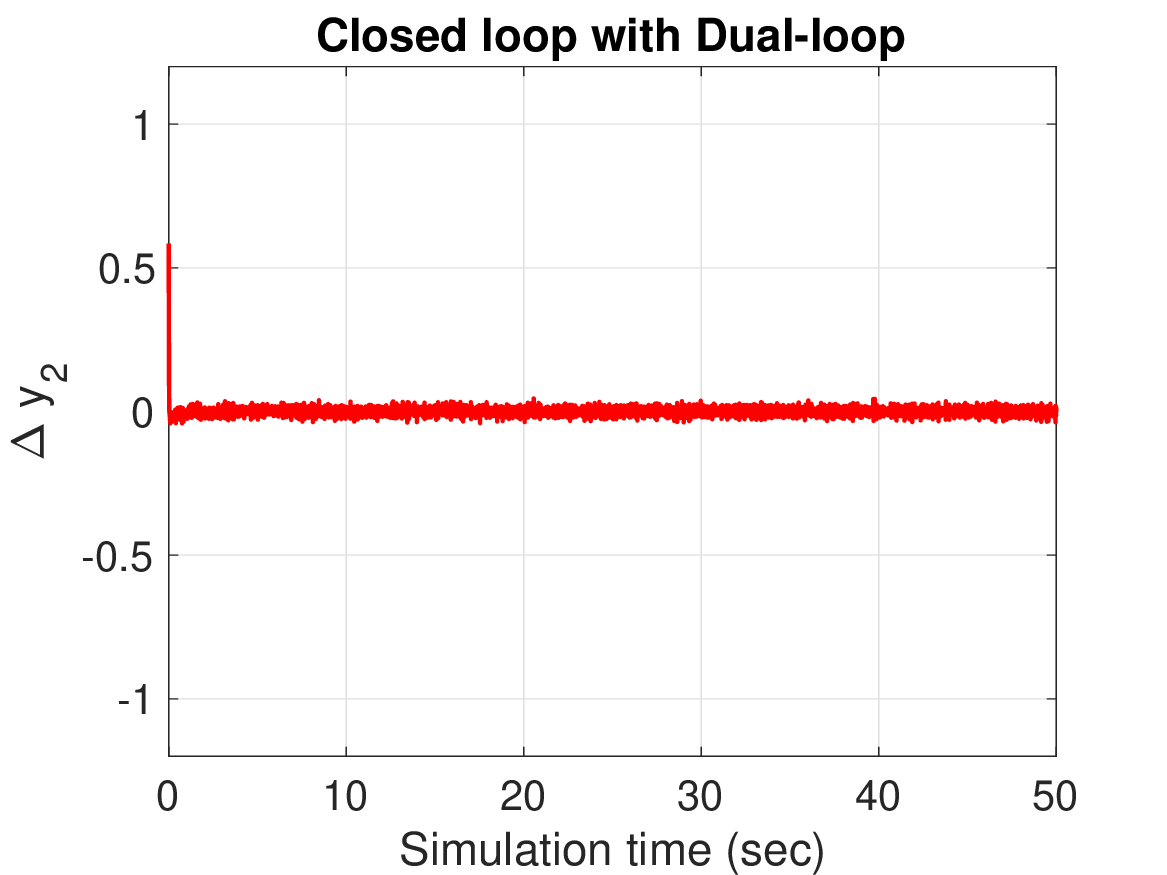}\label{figure2b}
\endminipage}
\caption{$\Delta$Y with dual-loop control at \textit{noise standard deviation} = 0.01}
\label{fig:DeltaY_DualLoop_n01}
\end{figure*}

\section{Conclusion}\label{section6}
The Koopman operator theory presents the opportunity to design control strategies for nonlinear complex dynamic systems using simple linear control theory. However, truncation error and identification biasedness due to measurement noise is a bigger roadblock to make it implementable for practical systems. The current work proposes to eliminate this gap by proposing a robust data-driven dual-loop control framework for nonlinear dynamical systems. The Koopman operator is used to identify high-dimensional approximate linear representation of the nonlinear dynamics, and the control strategy is developed based on the identified model. For learning the linear representation, Extended dynamical model decomposition (EDMD) is used. The study focused on the case of biased system identification due to measurement noise and how to design the control strategy to ensure stable and bounded system performance even with biased identified linear dynamics from Koopman operator theory. 

The paper first presents how measurement noise affects the identification process and then calculates the upper bound due to measurement noise. The identified system dynamics is then utilized to generate dual-loop control. The dual-loop framework contains Linear Quadratic Gaussian (LQG) as a stabilizing controller and $H_{\infty}$ control for guaranteeing performance robustness in the presence of model-mismatch. Based on truncation error due to approximation error and model biasedness due to measurement noise, a linear matrix inequality (LMI) condition is defined that ensures robust stability. To demonstrate the efficacy of the proposed approach, a simple Van Der Pol oscillator dynamics with control is considered. The results demonstrate that in the presence of both truncation error and model bias, a standard LQG control cannot regulate the system to the equilibrium point. However, our proposed dual-loop control robustly pushes the system to achieve equilibrium while minimizing output error between the estimated and actual system.


%





\ifCLASSOPTIONcaptionsoff
  \newpage
\fi



\bibliographystyle{IEEEtran}
\bibliography{IEEEabrv,bibtex/bib/References}

\end{document}